\documentclass[DIV=calc,paper=a4,fontsize=11pt,twocolumn]{scrartcl} 

\usepackage[english]{babel}
\usepackage{amsmath,amsfonts,amsthm}
\usepackage[final]{graphicx}
\usepackage{xcolor}
\usepackage[normal,small,hypcap,up,labelfont=bf,textfont=it]{caption}
\usepackage{epstopdf}
\usepackage{subfig}
\usepackage{booktabs}
\usepackage{fix-cm}
\usepackage{amssymb,amsfonts}
\usepackage{dsfont}
\usepackage{bbm}
\usepackage{pstricks}
\usepackage{cite}
\usepackage[utf8]{inputenc}
\usepackage[perpage,symbol*]{footmisc}
\usepackage[varg]{txfonts}
\usepackage{balance}
\usepackage{fancyhdr}
\PassOptionsToPackage{hyphens}{url}\usepackage[pdfencoding=auto,psdextra]{hyperref}
\usepackage{bookmark}
\usepackage{verbatim}
\usepackage{fontenc}
\usepackage{cuted}
\usepackage{widetext}

\DeclareCaptionFont{mycolor}{\color[HTML]{000000}}
\usepackage{sectsty}													
\allsectionsfont{
\color[HTML]{31ADF3}\usefont{OT1}{phv}{b}{n}
}

\sectionfont{
\color[HTML]{31ADF3}\usefont{OT1}{phv}{b}{n}
}

\usepackage{fancyhdr}												
\usepackage{lettrine}
\newcommand{\initial}[1]{%
\lettrine[lines=3,lhang=0.3,nindent=0em]{
\color[HTML]{31ADF3}
{\textsf{#1}}}{}}

\usepackage{setspace}
\usepackage{amscd,amssymb,latexsym}
\usepackage[mathscr]{euscript}
\usepackage{epsfig}
\usepackage{bm} 
\allowdisplaybreaks 

\usepackage{physics} 

\usepackage[linguistics]{forest} 

\newcommand{\N}{\mathbb N}
\newcommand{\R}{\mathbb R}

\newcommand{\tTr}{\text{Tr}\,} 
\newcommand{\tDim}{\text{dim}\,} 
\newcommand{\tL}{\text{L}\,} 
\newcommand{\tT}{\text{T}\,} 
\newcommand{\tD}{\text{D}\,} 

\newcommand{\cH}{\mathcal H}
\newcommand{\cC}{\mathcal C}
\newcommand{\cD}{\mathcal D}
\newcommand{\cB}{\mathcal B}

\newcommand{\cA}{\mathcal A}
\newcommand{\cX}{\mathcal X}
\newcommand{\cY}{\mathcal Y}
\newcommand{\cZ}{\mathcal Z}

\newcommand{\cF}{\mathcal F}
\newcommand{\cT}{\mathcal T}
\newcommand{\cE}{\mathcal E}
\newcommand{\cN}{\mathcal N}
\newcommand{\ch}{\text{H}} 

\newcommand{\lis}{\displaystyle (\lambda_{i})_{i=0}^{\infty}} 
\newcommand{\li}{ \lambda_i} 
\newcommand{\lin}{\lambda_{i}^{n}} 
\newcommand{\lins}{ (\lambda_{i}^{n})_{i=0}^{\infty}} 
\newcommand{\sumi}{\displaystyle \sum_{i=0}^{\infty}} 

\newcommand{\psis}{\{\psi_i\}_{i=0}^{\infty}} 
\newcommand{\psii}{\psi_i} 
\newcommand{\psin}{\psi_{i}^n} 
\newcommand{\psins}{\{\psi_{i}^{n}\}_{i=0}^\infty}
\newcommand{\sumii}{\displaystyle \sum_{i=0}^{I}} 
\newcommand{\tail}{\sum_{i=I+1}^{\infty}} 
\newcommand{\outeri}{\ket{\psi_i}\bra{\psi_i}} 
\newcommand{\outern}{\ket{\psi_{i}^n}\bra{\psi_{i}^n}} 

\newcommand{\geo}{\text{E}}
\newcommand{\ent}{\text{S}}

\newtheorem{Thm}{Theorem}[section]
\newtheorem{Cor}[Thm]{Corollary}
\newtheorem{Lem}[Thm]{Lemma}
\newtheorem{Prop}[Thm]{Proposition}
\newtheorem{Rmk}[Thm]{Remark}

\usepackage{titling}															

\newcommand{\HorRule}{\color[HTML]{31ADF3}
\rule{\linewidth}{1pt}%
}

\pretitle{\vspace{-30pt} \begin{flushleft} \HorRule
\fontsize{34}{34} \usefont{OT1}{phv}{b}{n} \color[HTML]{31ADF3} \selectfont
}
\title{Some Remarks on the Entanglement Number}					
\posttitle{\par\end{flushleft}\vskip 0.5em}

\preauthor{\begin{flushleft}\large \lineskip 0.5em \usefont{OT1}{phv}{b}{sl} \color[HTML]{31ADF3}}
\author{George Androulakis$^{~\mathsf{1}}$ \& Ryan McGaha$^{~\mathsf{2}}$\\[8pt]}											
\postauthor{\footnotesize \usefont{OT1}{phv}{m}{sl} \color[HTML]{000000}
$^{\mathsf{1}}$ Department of Mathematics, University of South Carolina, 
Columbia, South Carolina, USA. E-mail: \href{mailto:giorgis@math.sc.edu}{giorgis@math.sc.edu}\\
$^{\mathsf{2}}$ Department of Mathematics, University of South Carolina, 
Columbia, South Carolina, USA. E-mail: \href{mailto:rmcgaha@email.sc.edu}{rmcgaha@email.sc.edu}\\[10pt]		
\scriptsize\usefont{OT1}{phv}{m}{n} \color[HTML]{31ADF3}{\textbf{Editors: \emph{Stan Gudder} \& \emph{Danko Georgiev}} }\\[5pt]
\color[HTML]{000000}{Article history: Submitted on December 3, 2020; Accepted on December 7, 2020; Published on December 11, 2020.}
\par\end{flushleft}\HorRule}

\date{}																				

\begin{document}
\maketitle
\thispagestyle{fancy} 			
\initial{G}\textbf{udder, in a recent paper, defined a candidate entanglement measure which is called the entanglement number. The entanglement number is first defined on pure states and then it extends to mixed states by the convex roof construction. In Gudder's article it was left as an open problem to show that Optimal Pure State Ensembles (OPSE) exist for the convex roof extension of the entanglement number from pure to mixed states. We answer Gudder's question in the affirmative, and therefore we obtain that the entanglement number vanishes only on the separable states. More generally we show that OPSE exist for the convex roof extension of any function that is norm continuous on the pure states of a finite dimensional Hilbert space. Further we prove that the entanglement number is an LOCC monotone, (and thus an entanglement measure), by using a criterion that was developed by Vidal in 2000. We present a simplified proof of Vidal's result where moreover we use an interesting point of view of tree representations for LOCC communications. Lastly, we generalize Gudder's entanglement number by producing a monotonic family of entanglement measures which converge in a natural way to the entropy of entanglement.\\ Quanta 2020; 9: 22--36.}

\begin{figure}[b!]
\rule{245 pt}{0.5 pt}\\[3pt]
\raisebox{-0.2\height}{\includegraphics[width=5mm]{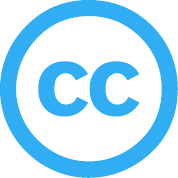}}\raisebox{-0.2\height}{\includegraphics[width=5mm]{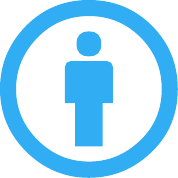}}
\footnotesize{This is an open access article distributed under the terms of the Creative Commons Attribution License \href{http://creativecommons.org/licenses/by/3.0/}{CC-BY-3.0}, which permits unrestricted use, distribution, and reproduction in any medium, provided the original author and source are credited.}
\end{figure}

\section{Introduction}

Entanglement is a fundamental concept in quantum mechanics. A detailed review article 
\cite{HHHH} summarized our understanding of entanglement until 2009, but many new results have since appeared. Gudder in \cite{G19} 
introduced a general theory of entanglement which applies even to classical probability measures with discrete 
support. In the same article, he also introduced a quantity that he called the \emph{entanglement number} which 
quantifies entanglement for classical probability measures as well as for density operators. Indeed the formula 
$e(\ket{\psi})=\sqrt{1-\sum_i \lambda_i^2} $ (where $(\sqrt{\lambda_i})_i$ are the Schmidt coefficients of the 
pure state $\psi$, with $\lambda_i \geq 0$ and $\sum_i \lambda_i=1$) defines the entanglement number 
$e(\ket{\psi})$ of a pure bipartite state $\ket{\psi}$ and can also be naturally extended to classical probability 
measures. The advantage of the entanglement number, besides the fact that a similar formula makes sense even in classical probability theory,
is that it can be computed easily for pure bipartite states. Indeed, in \cite[Theorem 4.2]{G20} a closed form 
formula is given for computing the entanglement number of a pure bipartite state $\ket{\psi}$. For example, if 
$\cX$, $\cY$ are two finite dimensional Hilbert spaces with orthonormal bases $(\ket{x_i})$ and $(\ket{y_j})$ 
respectively, and
$\ket{\psi} \in \cX \otimes \cY$ is a pure state written as $\ket{\psi} = \sum_{i,j} c_{i,j} \ket{x_i} \ket{y_j}$, 
then the entanglement number $e(\ket{\psi})$ can be computed by 
\begin{equation}\label{E:entanglement_number_for_pure_states_I}
e(\ket{\psi})= \sqrt{ 1- \tTr (|C|^4)}
\end{equation}
where $C$ is the matrix $(c_{i,j})$. Moreover, it is shown that if $(c_j)$'s are the columns of $C$ then 
\begin{equation}\label{E:entanglement_number_for_pure_states_II}
\tTr (|C|^4)= \sum_{r,s} |\langle c_r, c_s \rangle |^2.
\end{equation}

The extension of the entanglement number from pure states to mixed states is done via the \emph{convex roof} construction:

\begin{widetext}
\begin{equation}\label{E:entanglement_number_for_mixed_states}
 e(\rho) = \inf \left\{ \sumi \li e(\psi_i) : \li \geq 0, \sum_i \li =1, \psi_i\text{ are pure states, and }\sumi \li\outeri = \rho \right\}
\end{equation}
\end{widetext}

\noindent Gudder states in \cite{G19} the open question of whether the infimum in the above convex roof construction is always attained. Indeed some of his results, (such as \mbox{\cite[Theorem 3.3]{G19}}), depend on that assumption.
Here, we answer this question in the affirmative. Hence some of his results are strengthened.
Moreover, our Theorem~\ref{Thm:Convex_Roof} is general enough that it applies to all convex roof extensions of
norm continuous functions on pure states of finite dimensional Hilbert spaces; e.g., our Theorem~\ref{Thm:Convex_Roof} applies to the 
\emph{entanglement of formation} which is an entanglement measure, (see \cite{PV05,VPRK}), whose definition is based on the convex roof 
construction and introduced in \cite{BBPSSW96, BDSW96}. 

Another situation where our Theorem~\ref{Thm:Convex_Roof} can be applied is 
in the definition of the \emph{mutual entropy of a channel with respect to a state} which is defined in \cite{CNT87} via 
a variational procedure using the relative entropy between a state and all its pure state decompositions. 
It has been shown \cite{Uhl} that the infimum in the extension of the entanglement
entropy is always attained. Our result, Theorem~\ref{Thm:Convex_Roof}, yields this as a special case. 

Our result also applies to the convex roof extended \emph{negativity} \cite{LCDJ}, the family of
\emph{concurrences} \cite{GG}, the \emph{geometric measure of entanglement} \cite{BL}, and lastly 
the \emph{entanglement number} \cite{G19}, thus verifying the existence of OPSE for many common entanglement measures.

We also prove that the entanglement number is an entanglement
measure. Recall that if $\cX$, $\cY$ are two finite dimensional Hilbert spaces and $\tD (\cH )$ denotes the set of 
all density operators on a Hilbert space $\cH$, then a function $\mu:\tD(\cH)\rightarrow\R_{\geq0} $ is an \emph{entanglement measure} \cite{PV05, V98} if the following three properties hold for all $\rho\in D(\cH)$:
\begin{enumerate}
 \item $\mu(\rho)=0$ if and only if $\rho$ is separable, (i.e.\ $\mu$ is \emph{faithful}).
 \item $\mu(\Lambda(\rho))\leq \mu(\rho)$ where $\Lambda$ is an LOCC channel, (\i.e.~$\mu$~is an \emph{LOCC monotone}). LOCC stands for Local Operations and Classical Communication. See Section~\ref{Sec:LOCCmonotones} for a precise description of these channels. 
 \item $\mu((U_1 \otimes U_2) \rho (U_1^*\otimes U_2^*))=\mu(\rho)$ where $U_1$ and $U_2$ are unitaries, (i.e. $\mu$ is \emph{invariant under local unitaries}).
\end{enumerate}

A function $\mu$ on the states of a Hilbert space is called an \emph{LOCC monotone} \cite{V98} if property $2.$ is
satisfied. LOCC channels are hard to represent, \cite{LOCC}. In Section~\ref{Sec:LOCCmonotones} we provide 
a tree representation of LOCC channels which 
contributes to a better understanding of these channels. Our result, Theorem~\ref{Thm:LOCC_monotone} which is 
a slight strengthening of a result of 
Vidal~\cite{V98} with a simplified proof, shows that the entanglement number is an LOCC monotone.
Therefore we obtain that the entanglement number is indeed an entanglement measure. Furthermore we provide a family
of entanglement measures which are defined in a similar manner as the entanglement measure. These entanglement
measures are parametrized with $p \in (1,\infty)$ and they are equal to the entanglement number when $p=2$.

We would like to thank the editors for pointing out that Gudder had answered
the question that he posted in \cite{G19} in the publication \cite{G20b}.
We also would like to thank Shirokov who made us aware of his paper \cite{Shirokov10} where an extension 
of our Theorem~\ref{Thm:Convex_Roof} has been proved that is applicable even for infinite dimensional Hilbert
spaces. We were not aware of publications \cite{G20b} and \cite{Shirokov10}, and the 
 method used here in the proof of Theorem~\ref{Thm:Convex_Roof} is different that the one used in these publications. We decided to keep our proof of Theorem~\ref{Thm:Convex_Roof} for the completeness of the article.

\section{Existence of Optimal Pure State Ensembles}

In this section we prove that the infimum in the definition of convex roof construction which extends Gudder's 
entanglement number from pure states to mixed states, is always achieved. We present the argument in a more general
way so it can be applied to other convex roof constructions, (see Section~\ref{Sec:applications} for some applications). 

Assume that $\cX$ is a Hilbert space and $\mu$ is a function defined on the pure states of 
$\cH$ taking values in $\R_{\geq 0}$. One can extend $\mu$ to the set of density operators $\tD(\cX)$ of a 
Hilbert space $\cX$ by the convex roof construction:
\begin{equation} \label{E:convexroof}
 \mu(\rho) = \inf 
 \left\{ \sumi \li \mu(\psi_i) : \{ \lis, \{\psi_i\}_{i=0}^\infty \} \in \cC\cD (\rho)
 \right\}
\end{equation}
where the set $\cC\cD (\rho)$ of \emph{convex decompositions} of $\rho$ is defined as follows:

\begin{widetext}
\begin{equation}
\cC\cD (\rho) = \left\{ \{ \lis, \{\psi_i\}_{i=0}^\infty \} : \li \geq 0, \sum_{i=1}^\infty \li =1, \psi_i \in \cH, \text{ and } \sumi \li\outeri = \rho \right\}
\end{equation}
\end{widetext}

The main result of this section is the following:

\begin{Thm} \label{Thm:Convex_Roof}
If $\mu$ is a norm-continuous function on the pure states of a finite dimensional Hilbert space $\cH$, 
then the infimum in Equation~(\ref{E:convexroof}) is attained for all $\rho \in \tD(\cH)$. 
\end{Thm}

\begin{proof}
An important role in the proof of this theorem will be played by a compact metrizable topological space 
$(\Pi, \tau)$ that we define now.
Throughout the proof, for any Hilbert space $\cH$, we will denote by $S(\cH)$ the \emph{unit sphere} of $\cH$, 
(i.e. the vectors of $\cH$ of norm equal to $1$). As usual, the set $S(\cH)$ is identified with the set of pure
states of $\cH$.
Define 
\begin{equation}
 P_\infty =\{\lis : \li\geq0 \text{ for each i }, \sumi\li=1 \}
\end{equation}
i.e., $P_\infty$ is the intersection of $S(\ell^1)$ with positive cone of $\ell^1 $.
Next we define
\begin{equation}
 \Pi=P_\infty \times S(\cH)^{\N}
\end{equation}
 which intuitively can be thought of as the collection of all possible pairings of ``convex coefficients" with sequences of pure states. Equip $\Pi$ with the product topology $\tau$ of the weak$^*$ topology on $\ell_1$ with the infinite product topology of (any) norm topology of the unit sphere $S(\cH)$ of the Hilbert space $\cH$. 
 
Notice that $P_\infty$ is compact since it is a closed subspace of a weak$^*$ compact space. And as long as \mbox{$\tDim(\cH)<\infty$}, $S(\cH)$ is norm compact. This implies that space of all sequences of pure states, $S(\cH)^{\N}$ is norm compact by Tychonoff's theorem. Thus $\Pi$ is the product of two compact spaces and must therefore be compact. Note that we had to restrict ourselves to finite
dimensional Hilbert spaces since in general the set of pure states is not weak$^*$ compact, (see for example \cite[Theorem~2.8]{G60}).

Also notice that $(\Pi, \tau)$ is metrizable. Indeed the weak$^*$ topology on bounded subsets of $\ell_1$ is metrizable 
since the predual of $\ell_1$, (which is usually denoted by $c_0$), is separable. Also $S(\cH)^\N$ is the 
countable product of metric spaces, hence its product topology is metrizable. Thus $(\Pi, \tau)$ must be metrizable.

We split the rest of the proof into two Claims in order to make it more readable.

\noindent
\textbf{Claim 1:}
If $\cH$ is a finite dimensional Hilbert space, then for every $\rho \in \tD (\cH)$, the set $\cC\cD (\rho)$ is $\tau$-compact.

In order to prove Claim~1, assume that $\cH$ is a finite dimensional Hilbert space and define $f:\Pi\rightarrow\tD(\cH)$ by 
\begin{equation}
(\lis,\psis)\mapsto\sumi\li\outeri
\end{equation}
Since the $\lambda_i$'s are summable and the target space of $f$, (i.e. $\cD(\cH)$), is complete, this series 
will always converge.
Moreover notice that $f^{-1}(\{\rho\})$ is the set of all possible convex decompositions of $\rho$ into pure states, i.e. $f^{-1}(\{ \rho \})= \cC\cD(\rho)$. Since $f^{-1}(\{ \rho \})$ is a subset of the compact space
$\Pi$, 
in order to show that $f^{-1}(\{\rho\})$ is compact, it is enough to show that $f$ is continuous on $\Pi$. 
Since both spaces, $\Pi$ and $\cD (\cH)$ are metrizable, we will use sequences to check the continuity of $f$,
i.e. we will prove that $\pi_n\rightarrow\pi$ in $\Pi$ implies that $f(\pi_n)\rightarrow f(\pi)$ in $\tD(\cH)$.

So let $\pi_n =(\lins,\psins)$ be a sequence in $\Pi$ converging to $\pi=(\lis,\psis)$ as n goes to infinity, and let $\epsilon >0$. Then there exists some $I\in\N$ such that $\displaystyle \sum_{i=I+1}^\infty \li <\epsilon$, and so $\sumii \li \geq 1-\epsilon$. Moreover, since $\lins\rightarrow\lis$ in $\ell^1$, there exists some $N_0 \in \N$ such that $||(\lin)_{i=0}^I - (\li)_{i=0}^I ||_{\ell^1} < \epsilon$ for all $n\geq N_0$. And so $\sumii \lin\geq 1- 2\epsilon$ which yields that $\displaystyle\sum_{i=I+1}^\infty \lin <2\epsilon$. Thus we have the following inequalities:

\begin{widetext}
\begin{align}
 \nonumber ||f(\pi_n)-f(\pi)||_1 
 & \leq \left\| \sumii\lin\outern -\sumii\li\outeri\right\|_1 + 
 \left\|\tail \lin\outern - \tail\li\outeri\right\|_1\\
 \nonumber &\leq \left\| \sumii\lin\outern -\sumii\li\outeri \right\|_1 + \tail\lin \Bigg\|\outern \Bigg\|_1 + \tail\li \Bigg\|\outeri\Bigg\|_1 \\
 &=\left\| \sumii\lin\outern -\sumii\li\outeri\right\|_1 + \tail\lin + \tail\li\\
 \nonumber&< \left\|\sumii\lin\outern -\sumii\li\outeri\right\|_1 + 3\epsilon\\
 \nonumber&=\left\| \sumii (\lin -\li + \li)\outern -\sumii\li\outeri \right\|_1+3\epsilon\\
 \nonumber&\leq \sumii |\lin-\li| \cdot \Bigg\|\outern\Bigg\|_1 + \sumii\li \Bigg\|\outern-\outeri\Bigg\|_1 + 3\epsilon\\
 \nonumber&\leq\sumii \left(|\lin-\li| +\Bigg\|\outern-\outeri\Bigg\|_1\right) +3\epsilon
\end{align}
\end{widetext}

\noindent Now since weak-$^*$ convergence implies coordinate-wise convergence in $\ell^1$, we have that for all $n$ large enough, $|\lin-\li|<\frac{\epsilon}{I+1}$ for each $i$ $\in \{0, 1, \ldots , I\}$. We can also make sure that for all such $n$ and $i$, $\|\outern-\outeri\|_1 <\frac{\epsilon}{I+1}$, since $\psin$ converges to $\psii$. Thus for large enough $n$, we get that $\|f(\pi_n)-f(\pi)\|_1 <5\epsilon$, which implies that $f$ is continuous.
This finishes the proof of Claim~1.

Next we define $g:\Pi\rightarrow\R_{\geq0}$ by 
\begin{equation}
 (\lis,\psis)\mapsto\sumi\li \mu(\psi_i)
\end{equation}

Thus the statement that the infimum in Equation~(\ref{E:convexroof}) is always achieved, is equivalent to the fact that $g$ always attains its infimum on the set $\cC\cD (\rho)$, for any density matrix $\rho$. Since by Claim~1, $\cC\cD (\rho)$ is always $\tau$-compact, the proof of Theorem~\ref{Thm:Convex_Roof} will be complete once we prove the following:

\noindent
\textbf{Claim 2:} If $\mu$ is a norm-continuous function on the set of pure states of a finite dimensional
Hilbert space, then $g$ is continuous on $(\Pi , \tau)$.

The argument for the continuity of $g$ will be almost exactly the same as that for $f$, that is given in the 
Claim~1. Since $\mu$ is norm-continuous on the unit sphere of a finite dimensional Hilbert space which is 
norm-compact, we have
that $\mu$ is bounded, there exists some finite number $M$ such that $\mu(\psi)\leq M$ for all pure states $\psi$. 
Again, let $\pi_n =(\lins,\psins)$ be a sequence in $\Pi$ converging to $\pi=(\lis,\psis)$ as $n$ goes to infinity, and let $\epsilon >0$. Then using the same choice of $I$ as in the proof of the continuity of $f$, we have the following inequalities:

\begin{widetext}
\begin{align}
 \begin{split}
 |g(\pi)-g(\pi_n)|&\leq 
 \left|\sumii\li \mu(\psii)-\sumii\lin \mu(\psin)\right| +
 \left|\tail\li \mu(\psii)\right|+\left|\tail\lin \mu(\psin)\right|\\
 &\leq \left|\sumii\li \mu(\psii)-\sumii\lin \mu(\psin)\right| +\tail\li M+\tail\lin M \\
 &\leq\sumii |(\lin-\li+\li)\mu(\psin)-\li \mu (\psii)| + 3\epsilon M\\
 &\leq \sumii |\lin-\li|\mu(\psin) +\sumii\li |\mu(\psin)-\mu(\psii)| + 3\epsilon M\\
 &\leq \sumii |\lin-\li|M+\sumii|\mu(\psin)-\mu(\psii)| + 3\epsilon M\\
 \end{split}
\end{align}
\end{widetext}

Now taking $n$ large enough so that $|\lin-\li|<\frac{\epsilon}{I+1}$ and $|\mu (\psin)-\mu (\psii)|<\frac{\epsilon M}{I+1}$ for each $i$ $\in\{ 0,1, \ldots , I \} $, we get that $|g(\pi)-g(\pi_n)|<5\epsilon M$. Thus $g$ is also continuous of $\Pi$.
This finishes the proof of Claim~2 and thus the proof of Theorem~\ref{Thm:Convex_Roof}.
\end{proof}

\noindent
\textbf{Terminology:} The decomposition of a density operator $\rho$ with respect to a convex roof 
extension of an entanglement measure $\mu$ which is initially defined on the pure states of a multipartite 
Hilbert space is usually called an \emph{Optimal Pure State Ensemble}, which is often abbreviated as OPSE \cite{V98,Uhl}.

Theorem~\ref{Thm:Convex_Roof} immediately answers the question that Gudder posed in \cite{G19}:

\begin{Cor} \label{Cor:e_OPSE}
Let $\cH = \cX \otimes \cY$ be a finite dimensional bipartite Hilbert space and consider the entanglement number
defined on pure states of $\cH$ via Equations~(\ref{E:entanglement_number_for_pure_states_I}) and 
(\ref{E:entanglement_number_for_pure_states_II}), and extended to mixed states via 
Equation~(\ref{E:entanglement_number_for_mixed_states}). Then every mixed state $\rho \in \tD (\cH)$ admits 
an OPSE.
\end{Cor}

\begin{proof}
It is obvious that Equations~(\ref{E:entanglement_number_for_pure_states_I}) and 
(\ref{E:entanglement_number_for_pure_states_II}) define a norm-continuous function $e$ on the pure states of
$\cH$. Hence the result follows from Theorem~\ref{Thm:Convex_Roof}.
\end{proof}

\begin{Cor} \label{Cor:faithful}
Let $\cH$ be a finite dimensional multipartite Hilbert space 
i.e. $\cH$ is the tensor product of other Hilbert spaces and we implicitly refer to this decomposition when we mention factorable pure states of $\cH$ or separable states of $\cH$.
Let $\mu$ be a norm-continuous function on the pure states of $\cH$ with values 
in $\R_{\geq 0}$ which is 
faithful, in the sense that it only vanishes on the factorable pure states. Extend $\mu$ on the set of all states of $\cH$ via the convex roof construction. Then $\mu$ is faithful, i.e. for any density matrix $\rho$, we have 
$\mu(\rho)=0$ if and only if $\rho$ is separable. 
\end{Cor}

\begin{proof}
$\implies)$ Suppose that $\rho \in \tD(\cH)$ with $\mu(\rho)=0$. Then by Theorem~\ref{Thm:Convex_Roof}, 
there exist some \mbox{$\{ (\lambda)_i, (\psi_i) \} \in \cC\cD (\rho)$} such that 
\begin{equation}
\displaystyle\sum_i \lambda_i \mu(\psi_i) = \mu(\rho).
\end{equation}
Since $\mu$ only takes non-negative values, we obtain that $\mu(\psi_i)=0$ for each $i$. Since $\mu$ is assumed 
to be faithful on the pure states, this implies that each $\psi_i$ is factorable, and so $\rho$ is the convex 
combination of some factorable pure states and is therefore separable. 

$\impliedby)$ Conversely, suppose that $\rho$ is separable. Then $\rho$~is the convex combination of some factorable pure states in $\cH$. Let $\rho= \sum_i \lambda_i \psi_i$ where $\lambda_i$'s are convex coefficients 
and $\psi_i$'s are factorable pure states of $\cH$. Since $\mu$ is faithful on the pure states of $\cH$, we have that $\mu (\psi_i)=0$ for every~$i$. Then 
\begin{equation}
\sum_i \lambda_i \mu(\psi_i) =0 
\end{equation}
which, by the definition of the convex roof extension, implies that $\mu(\rho) =0$ since each $\mu$ takes on nonnegative values 
for each $i$. 
\end{proof}

\section{Other Applications to Convex Roof Extensions} \label{Sec:applications}

Now that we have shown the existence of OPSE for general norm-continuous functions on pure states of finite dimensional Hilbert spaces 
we can present some applications besides Corollary~\ref{Cor:e_OPSE}. In literature the existence of OPSE has many times
been taken for granted, but other times it has been questioned \cite{G19}, and other times claimed in a particular 
special case \cite{Uhl}. 

Recall that for Hilbert spaces $\cA$ and $\cB$ the \emph{entanglement entropy} $\ent$ of a pure state $\ket\psi\in \cA\otimes\cB$ is defined by

\begin{widetext}
\begin{equation}
\text{S}(\psi)=-\tTr\left( \tTr_\cA(\ket\psi\bra\psi) \log (\tTr_\cA \ket\psi\bra\psi) \right) =
-\tTr\left( \tTr_\cB(\ket\psi\bra\psi) \log (\tTr_\cB \ket\psi\bra\psi) \right)
\end{equation}
\end{widetext}

\noindent The entanglement entropy is then extended by the convex roof construction as in Equation~(\ref{E:convexroof}) by defining 
\begin{equation}
\ent(\rho) = \inf 
 \left\{ \sumi \li \ent(\psi_i) : \{ \lis, \{\psi_i\}_{i=0}^\infty \} \in \cC\cD (\rho)
 \right\} 
\end{equation}
for all density operators $\rho\in\tD(\cA\otimes\cB)$. The convex roof extension of the entanglement entropy
is also often called the $\emph{entanglement of formation}$. It has been shown that the von Neumann entropy,
$\text{H}(\rho)=-\tTr\rho\log\rho$, is a continuous mapping on density operators \cite{FAN}. Moreover 
it is well known that the
partial trace is a bounded linear operator from the space of trace class operators on a bipartite Hilbert
space to the space of trace class operators of one of its parts, and therefore continuous.
Thus the entanglement entropy is continuous on pure states since it is the composition of two continuous
maps, and therefore, by Theorem~\ref{Thm:Convex_Roof}, it must always exhibit OPSE. 

Next we apply our result to the convex roof extended \emph{negativity} \cite{LCDJ,VW} of a state. Recall that the negativity \cite{VW} of a state $\rho \in \tD(\cA\otimes\cB)$ is defined by 
\begin{equation}
\cN(\rho)=\frac{1}{d-1} (\| \rho ^{T_\cB} \| -1 )
\end{equation}
where $d=\min\{ \tDim\cA , \tDim\cB \}$, $\rho^{T_\cB}$ is the partial transpose of $\rho$ on the space $\tD(\cB)$, and the norm is the Hilbert--Schmidt norm on operators on $\cA\otimes\cB$. 
While $\cN$ is indeed defined for all density operators, it cannot distinguish bound entangled states 
and separable states. One solution \cite{LCDJ} to this problem is to first define the negativity for 
pure states $\ket\psi\in\cA\otimes\cB$ by \begin{equation}
\cN(\psi)=\frac{1}{d-1} (\| \ket\psi\bra\psi ^{T_\cB} \| -1 )
\end{equation}
then extend the domain of $\cN$ by the convex roof construction. This definition of $\cN$ can then 
distinguish bound entangled and separable states and is an entanglement measure \cite{LCDJ}. Now since the 
norm of a Hilbert space is obviously a norm-continuous function, and partial transpose is also continuous, 
it follows that $\cN$ is continuous on pure states. Therefore, by Theorem~\ref{Thm:Convex_Roof}, 
the convex roof extended negativity of a state must admit OPSE. 

We can also apply our result to the family of concurrence monotones $C_k$ \cite{GG} of a bipartite pure state $\psi\in\cA\otimes\cB$ which are defined by 
\begin{equation}
C_k(\psi) = \left(\dfrac{\ent_k (\lambda_1,\dots,\lambda_d)}{S_k(\frac{1}{d},\dots,\frac{1}{d})} \right)^\frac{1}{k}
\end{equation}
where $\ent_k$ is the $k$-th elementary symmetric polynomial \cite{DE}, the $\lambda_j $ are the Schmidt 
coefficients of $\psi$, $d = \min\{\tDim\cA, \tDim\cB\} $, and $k$ ranges from 1 to $d$. The $C_k$'s are then extended to the set of all density operators on
\mbox{$\cA\otimes\cB$} by the convex roof construction as well. 
Recall that the $k$-th elementary symmetric polynomial of $d$ variables is defined by 
$\ent_k(x_1, \ldots , x_d)= \sum_{1\leq i_1<i_2<\ldots <i_k \leq d} x_{i_1}x_{i_2}\ldots x_{i_k}$.
Since the $C_k$ are just the $1/k$-th powers of the normalized elementary symmetric polynomials 
in the Schmidt coefficients of a state, they are continuous on $\cA\otimes\cB$ for each~$k$. Thus, 
by Theorem~\ref{Thm:Convex_Roof}, the generalized concurrence monotones also admit OPSE. 

A particularly interesting example is the \emph{geometric measure of entanglement} \cite{BL}. Let $\cH=\cX_1\otimes \cdot \cdot\cdot \otimes \cH_n$ be the tensor product of some finite dimensional Hilbert spaces $H_i$. Then for pure states $\ket\psi\in\cH$, the \emph{geometric measure of entanglement} is a family of entanglement measures defined by 
\begin{equation}
\geo_{k_1,\dots,k_n}(\psi) = \sup_{P_1,\dots,P_n} ||P_1 \otimes \cdot\cdot\cdot\otimes P_n \ket\psi ||^2 
\end{equation} 
where each $P_j$ is a rank $k_j$ orthogonal projector on the space $H_j$ for each index $j$. The supremum in this
definition is achieved if $\cH$ is a finite dimensional Hilbert space, since the set of projections of a finite
dimensional Hilbert space is compact with the norm topology. This measure is an example of an 
\emph{increasing LOCC monotone}; i.e., \mbox{$\geo$~cannot} \emph{decrease} under LOCC channels. While at first 
glance such monotones seem to do the opposite of what we discuss in Section~\ref{Sec:LOCCmonotones} 
there is a one to one correspondence between the set increasing LOCC monotones and the set of (decreasing) 
LOCC monotones (the type discussed in Section~\ref{Sec:LOCCmonotones}). To see this relationship, consider 
the map $\mu\mapsto\sup_{\psi}\mu(\psi) - \mu$ for all increasing LOCC monotones~$\mu$. We leave 
it to the reader to verify that this map is indeed a bijection between the two types of non-negative and 
bounded monotones. Another difference between the two types of monotones is that increasing monotones are 
extended to general density operators by a concave roof construction:
\begin{equation}
\mu(\rho) = \sup \left\{ \sum_{i=1}^\infty \lambda_i\mu(\psi_i) : \{ (\lambda_i)_{i=1}^\infty , (\psi_i)_{i=1}^\infty \}\in \mathcal{CD}(\rho) \right\} .
\end{equation} 
But because of the correspondence between the two types of monotones, we have also shown the existence of OPSE for concave roof constructions as well. To apply our result to the geometric measure of entanglement, we must first show that it is continuous on the set of pure states of a general multipartite space $\cH$. 
\begin{Prop}
Let $\cH=H_1\otimes \cdot\cdot\cdot \otimes \cH_n$ be a tensor product of finite dimensional Hilbert space, then $\geo_{k_1 \dots k_n}$ is continuous for all choices of $k_j\leq \tDim(\cH_j)$ for each $j$. 
\end{Prop}
\begin{proof}
Let $(\psi_i)_{i=1}^\infty$ be a sequence of pure states in $\cH$ converging to $\psi$ for some pure state $\psi$. Then

\begin{widetext}
\begin{equation}
\begin{split}
 |\geo_{k_1\dots k_n}(\psi_i)-\geo_{k_1\dots k_n}(\psi) | &= | \sup_{P_1\dots P_n}||P_1\otimes \cdot\cdot\cdot\otimes P_n \ket{\psi_i}||^2 - ||Q_1\otimes \cdot\cdot\cdot \otimes Q_n \ket\psi ||^2 |\\
 &\leq \sup_{P_1\dots P_n}|\text{ } ||P_1\otimes \cdot\cdot\cdot \otimes P_n \ket{\psi_i}||^2 - ||P_1\otimes \cdot\cdot\cdot\otimes P_n \ket{\psi}||^2 |\\
\end{split}
\end{equation}
\end{widetext}

\noindent Now using the reverse triangle inequality and the fact that orthogonal projectors have operator norm 1, the right hand side of the last inequality is less than or equal to

\begin{widetext}
\begin{equation}
\begin{split}
 \sup_{P_1\dots P_n} &|\text{ } ||P_1\otimes \cdot\cdot\cdot \otimes P_n \ket{\psi_i}|| - ||P_1\otimes \cdot\cdot\cdot\otimes P_n \ket{\psi}|| \text{ }|\cdot |\text{ } ||P_1\otimes \cdot\cdot\cdot\otimes P_n \ket{\psi_i}|| + ||P_1\otimes \cdot\cdot\cdot\otimes P_n \ket{\psi}|| \text{ }|\\
 &\leq 2 \sup_{P_1\dots P_n} |\text{ } ||P_1\otimes \cdot\cdot\cdot \otimes P_n \ket{\psi_i}|| - ||P_1\otimes \cdot\cdot\cdot\otimes P_n \ket{\psi}|| \text{ }|\\
 &\leq 2\sup_{P_1\dots P_n} || P_1\otimes \cdot\cdot\cdot \otimes P_n(\ket{\psi_i}-\ket\psi) || \\
 &\leq 2\sup_{P_1\dots P_n} || P_1\otimes \cdot\cdot\cdot \otimes P_n ||\cdot || \psi_i-\psi ||\\
 & \leq 2 || \psi_i-\psi ||
\end{split}
\end{equation}
\end{widetext}

\noindent Now letting $i$ tend to infinity, it follows that \mbox{$|\geo_{k_1\dots k_n}(\psi_i)-\geo_{k_1\dots k_n}(\psi) |\rightarrow 0 $}, implying that $\geo_{k_1\dots k_n}$ is continuous for all possible combinations of ranks $k_1\dots k_n$ of local projections on $\cH$. 
\end{proof}
Thus Theorem~\ref{Thm:Convex_Roof} applies to the geometric measure of entanglement, and therefore this entanglement measure also exhibits OPSE. 

While convex roof constructions often appear in the discussion of entanglement measures, they also arise in other contexts in quantum information theory. For instance, the entropy of a channel $\Phi$ with respect to a state~$\rho$ \cite{OHYA} is defined by

\begin{widetext}
\begin{equation}
\ch_\rho (\Phi) = \ch\left( \Phi(\rho) \right) - \inf\left\{ \sum_{i=1}^\infty \lambda_i\ch\left(\Phi(\psi_i)\right) : \{ (\lambda_i)_{i=1}^\infty , (\psi_i)_{i=1}^\infty \}\in \mathcal{CD}(\rho) \right\} 
\end{equation}
\end{widetext}

\noindent where $\ch(\sigma)$ is the von Neumann entropy of a state $\sigma$. The existence of OPSE for this case was proven by Uhlmann \cite{Uhl} but our Theorem~\ref{Thm:Convex_Roof} is applicable to this problem as well. 
Since the von Neumann entropy is continuous \cite{FAN} and since all quantum channels are continuous
\cite{A17}, $\ch\left(\Phi(\psi)\right)$ is a continuous function on pure states $\psi$ of a Hilbert space. 
Thus $\ch_\rho(\Phi)$ exhibits OPSE for all states $\rho$ and fixed channels $\Phi$, verifying 
Uhlmann's result. 

\section{LOCC Channels and LOCC Monotones} \label{Sec:LOCCmonotones}

In this section we show that Gudder's entanglement number is an LOCC monotone using a slight extension of a 
criterion due to Vidal \cite{V98}. Moreover we simplify Vidal's proof of this criterion. Furthermore during the course of its proof, 
we provide a representation of LOCC operations using trees, which we believe gives a better understanding to the complicated notion of LOCC channels.

Vidal \cite{V98} shows the following result: Assume that $\cX$ is a finite dimensional Hilbert space, and 
\mbox{$f:\tD (\cX) \to \R_{\geq 0}$} is a function which is invariant under unitaries (i.e.~$f(U\rho U^*)=f(\rho)$
for every $\rho \in \tD (\cX)$ and every unitary operator $U$ on $\cX$), and concave (i.e.~$f(\lambda \sigma_1 + (1-\lambda) \sigma_2) \geq \lambda f(\sigma_1) + (1-\lambda) f(\sigma_2)$ for all $\lambda \in [0,1]$ and all 
$\sigma_1, \sigma_2 \in \tD (\cX)$). Then define a function $\mu$ on pure states of $\cX_1 \otimes \cX_2$ 
where $\cX_1=\cX_2=\cX$, by 
\begin{equation}\label{E:def}
 \mu(\psi) = f(\tTr_{\cX_1} \ket\psi \bra\psi) = f(\tTr_{\cX_2} \ket\psi\bra\psi)
\end{equation}
Extend the function $\mu$ from the pure states of $\cX_1 \otimes \cX_2$ to all states of 
$\cX_1 \otimes \cX_2$ via the 
convex roof construction. If the infimum in the definition of the convex roof is always attained, (i.e. if 
OPSE's exist for every mixed state), then the extension of $\mu$ to mixed states of $\cX_1 \otimes \cX_2$ via 
the convex roof construction is an LOCC monotone.

Recall that a function $\mu$ defined on density operators of a multipartite Hilbert space $\cH$ and taking non-negative real values
is called an LOCC monotone if $\mu (\Lambda (\rho)) \leq \mu (\rho)$ for all density operators 
$\rho \in \tD (\cH)$ and all
LOCC channels $\Lambda$. The image of an LOCC channel may not be the density operators on the tensor product of two identical Hilbert spaces,
and indeed it can be easily verified that Vidal's proof extends to that case very easily. Moreover,
it is well known that for any two finite dimensional Hilbert spaces $\cX_1$ and $\cX_2$, and for every normalized
vector $\psi$ of $\cX_1 \otimes \cX_2$, the set of non-zero eigenvalues of $\tTr_{\cX_1} \ket \psi \bra \psi $ is 
equal to the set of non-zero eigenvalues of $\tTr_{\cX_2} \ket \psi \bra \psi$. Thus if the Hilbert spaces 
$\cX_1$ and $\cX_2$ are equal, (or at least have equal dimension), then the matrices 
$\tT_{\cX_1}\ket\psi \bra\psi$ and $\tTr_{\cX_2}\ket\psi\bra\psi$
are unitarily equivalent. In the proof that we present below we do not assume that the Hilbert spaces 
$\cX_1$ and $\cX_2$ are equal, or have equal dimension, and we simply assume that $f$ is concave function that depends only on the nonzero eigenvalues of the densities matrices. Moreover, our proof is simpler than Vidal's 
proof, and along the proof we give a pictorial tree representation of LOCC channels that helps to understand
this notion. More precisely, our main result of this section is the following:

\begin{Thm} \label{Thm:LOCC_monotone}
Assume that a function $f$ is defined on the density operators of all finite dimensional Hilbert spaces and takes values in non-negative real numbers. 
Assume that $f$ is concave, and it depends only the non-zero eigenvalues of its argument. 
Let $\mu$ be defined on pure states of any bipartite Hilbert space via Equation~(\ref{E:def}), and extended
to all mixed states of the bipartite Hilbert space via Equation~(\ref{E:convexroof}), and assume that the infimum 
in (\ref{E:convexroof}) is always achieved. Then $\mu$ is an LOCC monotone. 
\end{Thm}

Before providing the proof of this result, we will provide a discussion on the structure of LOCC channels 
using the notion of the tree. 

Let $\cA$ be a Hilbert space whose states can be manipulated by only one party Alice, and $\cB$ be a Hilbert space whose states can be manipulated by only one party Bob (\emph{local operations}). Then Alice may perform quantum channels on her space which will be of the form $\Phi_\cA (X)=\sum_i A_iXA_i^*$ and Bob may perform channels on his space which may be of the form $\Phi_\cB(Y)=\sum_j B_jXB_j^{*}$ where $\sum_i A_i^* A = I_\cA$ and $\sum_j B_j^* B_j = I_\cB$. Note that for simplicity, we keep the domains and ranges of each channel vague and simply write $\cA$ to be the current space of Alice and $\cB$ to be the current space of Bob, even though these
spaces $\cA$ and $\cB$ keep changing during the application of every channel of the LOCC communication.
Now suppose that Alice and Bob share a state $\rho\in\tD(\cA\otimes\cB)$, then Alice could perform 
$\Phi_\cA$ and Bob could perform $\Phi_\cB$ and the post operation state would be 
$(\Phi_\cA\otimes\Phi_\cB)(\rho)$. But operations of this form do not explicitly allow for communication 
between the parties and therefore do not adequately describe LOCC channels. In order to incorporate 
classical communications, we must have that each party's channels depend on the other's in some way. 

Consider the following operation on a state \mbox{$\rho\in\tD(\cA\otimes\cB)$}:
Alice performs the channel $\Phi_\cA$ then measures the outcomes of her operation and sends the result to Bob via some method of classical communication. To an outside observer, the operation on~$\rho$ would appear as 
\begin{equation}
\Phi_\cA(\rho)=\sum_i \cE_{i} (\rho) \text{ where } \cE_{i}(\rho)= (A_i \otimes I_\cB)\rho (A_i^*\otimes I_\cB)
\end{equation}
The outcome states of such an operation would be given by

\begin{widetext}
\begin{equation}
\rho_i = \frac{\cE_i(\rho)}{\tTr(\cE_i(\rho))}=
\frac{(A_i\otimes I_\cA)\rho(A_i^*\otimes I_\cA)}{\tTr\left( (A_i^* A_i\otimes I_\cB )\rho \right)} 
\text{ with probability } \tTr\left( (A_i^* A_i\otimes I_\cB )\rho \right) \text{ for each } i.
\end{equation}
\end{widetext}

\noindent Alice then measures the state on her system and sends the result to Bob. Then Bob acts accordingly and applies the channel 
\begin{equation}
\Phi_i (\sigma)= \sum_j \cE_{i,j}(\sigma) \text{ where } 
\cE_{i,j}(\sigma)= (I_\cA \otimes B_{ij})\sigma(I_\cA \otimes B_{ij}^*)
\end{equation}
when he learns that the measured state was $\rho_i$. Since $\Phi_i$~is a channel for every fixed $i$, 
we obtain that, 
\begin{equation}
\sum_j B_{ij}^* B_{ij}=I_\cB 
\end{equation}
and the channel that describes the final outcome of the scenario above will be of the form 
\begin{align}
\Lambda(\rho) & =\sum_{i,j} (A_i\otimes B_{ij})\rho (A_i\otimes B_{ij}^*) \nonumber\\
& = \sum_i \sum_j \cE_{i,j} \circ \cE_i (\rho).
\end{align}

But not all LOCC channels will be this simple. This scenario was only one round of operations and the communication only went in one direction. To extend the definition to more rounds of communication going in various directions, we can iteratively extend our scenario using the following tree structure. 

Let $\cT\subset \{ \emptyset \} \cup \bigcup_{n=1}^{\infty} \N^n$ be a finite collection of finite ordered 
subsets of $\N$ including the empty set. Endow $\cT$ with a partial order $ \prec$ by defining 
$ (x_1,\dots,x_k)\prec (y_1,\dots,y_l)$ if $k<l$ and $x_i=y_i$ for all $i\in\{1,\dots,k\}$. 
Moreover we define $\emptyset \prec x$ for all $x\in\cT \backslash \{ \emptyset \}$, and we denote the
length $l(x)=k$ if $x=(x_1,\dots,x_k)$ for some $x_i\in\N$ for each $i$. Also define $l(\emptyset)=0$.
Lastly we denote the \emph{immediate successors} 
of an element $x$ by $I(x)=\{y\in\cT \big| x\prec y \text{ and } l(y)=l(x)+1 \}$ and we denote $\cF(\cT)$
to be the collection of \emph{final nodes} of the tree $\cT$ where the final nodes are the nodes which 
have no immediate successors. Using this notion of a tree, we can index all possible states that occur 
in an LOCC process by a tree $\cT$. We will also use the convention $\rho_{\emptyset}$ to represent $\rho$
before any operations have been applied. Thus an LOCC process can be visualized using the following tree

\begin{widetext}
\begin{center}
\scalebox{1.1}{
\begin{forest}
[$\rho_{\emptyset}$
 [$\rho_{ (1) }$
 [$\rho_{ (1,1) }$
 [$\cdot$]
 ]
 [$\cdot$]
 [$\rho_{ (1,K_{1,1}) }$
 [$\cdot$]
 ]
 ]
 [$\rho_{ (2) }$
 [ $\rho_{ (2,1) }$
 [$\cdot$]
 ]
 [$\cdot$]
 [ $\rho_{ (2,K_{1,2}) }$
 [$\cdot$]
 ]
 ]
 [$\cdot$]
 [$ \rho_{ (K_1 -1) }$
 [$\rho_{ (K_1 -1,1) }$
 [$\cdot$]
 ]
 [$\cdot$]
 [$\rho_{ (K_1-1,K_{1,K_1 -1}) }$
 [$\cdot$]
 ]
 ]
 [$\rho_{ (K_1) }$
 [$\rho_{ (K_1 -1,1) }$
 [$\cdot$]
 ]
 [$\cdot$]
 [$\rho_{ (K_1,K_{1,K_1})}$
 [$\cdot$]
 ]
 ]
]
\end{forest}
}
\end{center}
\end{widetext}

\noindent where $\rho_x$ represents the \emph{unnormalized} state that occurs after the $x$-th Krauss operator has been applied to previous \emph{normalized} state. Moreover, each row of the tree represents a single party applying conditionally operations while the other party does nothing. Hence, each node in the tree $\rho_y$ is of the form
\begin{align}
\rho_y &=(A_y\otimes I_\cB )\dfrac{\rho_x}{\tTr\rho_x}(A_y^* \otimes I_\cB)\text{ or } \rho_y \nonumber \\
& = (I_\cA \otimes B_y)\dfrac{\rho_x}{\tTr\rho_x} (I_\cA \otimes B_y^*)
\end{align} 
where $y\in I(x)$, $A_x$ is an operation on Alice's space, and $B_x$ is an operation on Bob's space. The physical significance of this construction being that $\dfrac{\rho_y}{\tTr\rho_y}$ is the state observed with probability $\tTr\rho_y$ after the previous party's operations. In order to refer to the Krauss operators at each node, we can define the maps $\cE_y$ by 
\begin{align}
\cE_y(X) & =(A_y\otimes I_\cB )X(A_y^* \otimes I_\cB)\text{ or } \cE_y(X) \nonumber\\
& = (I_\cA \otimes B_y)X(I_\cA \otimes B_y^*)
\end{align}
depending on which party is applying operations at node~$y$. 
Thus using this notation, we express an LOCC channel $\Lambda$ as 
\begin{equation}
\Lambda(\rho) = \sum_{y_1 \in I(\emptyset)}\cdot\cdot\cdot\sum_{y_n\in I(y_{n-1})} \cE_{y_n}\circ\cdot\cdot\cdot\circ\cE_{y_1}(\rho_\emptyset)
\end{equation}
where the $y_n$ are the final nodes of the tree $\cT$. Then to prove LOCC monotonicity, we need to ensure 
\begin{equation}
\mu\left( \sum_{y_1 \in I(\emptyset)}\cdot\cdot\cdot\sum_{y_n\in I(y_{n-1})} \cE_{y_n}\circ\cdot\cdot\cdot\circ\cE_{y_1}(\rho_\emptyset ) \right)\leq\mu(\rho_\emptyset).
\end{equation}
But before we prove this we will need the following lemma to ensure the self containment of the proof of our main result:
\begin{Lem}
Let $\cH=\cA\otimes\cB$ be a finite dimensional Hilbert space, then the partial traces $\tTr_\cA$ and $\tTr_\cB$ are cyclic on the spaces $\cA$ and $\cB$ respectively. 
\end{Lem}
\begin{proof}
Let $A_1\in\tL(\cA_1,\cA_2)$, $A_2\in\tL(\cA_2,\cA_1)$, $\cB_1\in\tL(\cB_1,\cB_2)$, and $\cB_2\in\tL(\cB_2,\cB_1)$, then using the cyclicity of the trace operator on $\cA$ we have that
\begin{align}
\tTr_{\cA_2}\left( A_1 A_2 \otimes B_1 B_2 \right)
& = (\tTr\otimes I_{\cB_2})\left( A_1 A_2 \otimes B_1 B_2 \right) \nonumber\\
& = \tTr \left(A_1 A_2\right) B_1 B_2 \nonumber\\
& = \tTr\left(A_2 A_1\right) B_1 B_2 \nonumber\\
& = (\tTr \otimes I_{\cB_2})\left( A_2 A_1 \otimes B_1 B_2 \right) \nonumber\\
& = \tTr_{\cA_1}\left( A_2 A_1 \otimes B_1 B_2 \right) 
\end{align}
which completes the proof.
\end{proof}
Henceforth, instead of writing the precise space that is being traced away, we will use the convention that $\tTr_\cA$ traces out the space on the left, no matter what Hilbert space $\cA$ is, and $\tTr_\cB$ traces out the Hilbert space on the right. With this convention, the previous lemma can be written as
\begin{equation}
\tTr_{\cA}\left( A_1 A_2 \otimes B_1 B_2 \right) = \tTr_{\cA}\left( A_2 A_1 \otimes B_1 B_2 \right)
\end{equation}

Using this lemma we can prove another useful property of the partial trace.

\begin{Cor} \label{Cor:cyclicity}
Let $A\in \tL(\cA_1,\cA_2)$, and $\rho\in\tD(\cA_1\otimes\cB)$, then 
\begin{equation}
\tTr_\cA \left((A\otimes I_\cB)\rho(A^* \otimes I_\cB) \right) = \tTr_\cA \left( (A^*A\otimes I_\cB)\rho\right) 
\end{equation}
\end{Cor}

\begin{widetext}
\begin{proof}
Let $\{\ket i\}_{i\in I}$ and $\{\ket j\}_{j\in J}$ be orthonormal bases for the Hilbert spaces $\cA_1$ and $\cB$ respectively. Then every $\rho\in \tD(\cA_1\otimes\cB) $ can be written as 
$\rho=\sum_{i,j,k,l} \rho_{i,j,k,l} \ket{i}\bra{k}\otimes \ket{j}\bra{l}$. Thus
\begin{align}
 \tTr_\cA \left( (A\otimes I_\cB) \rho (A^*\otimes I_\cB) \right) &= \sum_{i,j,k,l}\rho_{i,j,k,l}\tTr_\cA \left( (A\otimes I_\cB) (\ket{i}\bra{k} \otimes \ket{j}\bra{l})(A^*\otimes I_\cB)\right) \nonumber \\
 & =\sum_{i,j,k,l}\rho_{i,j,k,l}\tTr_\cA \left( A\ket{i}\bra{k}A^* \otimes \ket{j}\bra{l} \right) \nonumber \\
 &=\sum_{i,j,k,l}\rho_{i,j,k,l}\tTr_\cA \left( A^*A\ket{i}\bra{k} \otimes\ket{j}\bra{l}\right) 
 = \tTr_\cA \left( (A^*A\otimes I_\cB)\rho\right)
\end{align}
which is the desired result.
\end{proof}
Finally we are ready to give the 
\begin{proof}[Proof of Theorem~\ref{Thm:LOCC_monotone}]
Let $\rho\in\tD(\cA\otimes\cB)$ and let $\Lambda$ be an LOCC channel with its corresponding tree $\cT$ with root $\rho_\emptyset = \rho$. In order to show that $\mu$ decreases under $\Lambda$, we will first show that 
\begin{equation}
\mu \left(\frac{\rho_x}{\tTr\rho_x}\right) \geq \sum_{y\in I(x)} \tTr\rho_y\mu \left(\frac{\rho_y}{\tTr\rho_y}\right) 
\end{equation} 
where $x\in\cT\setminus\cF(\cT)$ and the $\rho_x$ and $\rho_y$ are the unnormalized states at nodes $x$ and $y$ 
respectively. Thus we need to show that $\mu$ decreases on average at each node in the tree. So let $x\in\cT$ be 
a non-final node and let $\rho_x$ be the unnormalized state at node $x$. Then without loss of generality we 
can write 
\begin{equation}
\rho_y= (A_y \otimes I_\cB)\dfrac{\rho_x}{\tTr\rho_x}(A_y^*\otimes I_\cB) 
\end{equation}
for each $y\in I(x)$, where 
\begin{equation} \label{E:SumToIdentity}
\sum_y A_y^* A_y = I_\cA.
\end{equation}
Thus the normalized state at node $y$ will be of the form $\frac{\rho_y}{\tTr(\rho_y)}$. Now let 
$\{ (\lambda_i)_i , (\psi_i)_i \}$ be an OPSE for $\dfrac{\rho_x}{\tTr\rho_x}$. Then
\begin{equation}
 \mu \left(\frac{\rho_x}{\tTr\rho_x} \right) =\sum_i \lambda_i \mu(\psi_i) 
 = \sum_i \lambda_i f\left( \tTr_\cA \ket{\psi_i} \bra{\psi_i} \right) 
 = \sum_i \lambda_i f\left( \tTr_\cA\sum_{y\in I(x)} (A_y\otimes I_\cB)\ket{\psi_i} \bra{\psi_i}(A_y^*\otimes I_\cB) \right) 
\end{equation}
where the last equality follows from Corollary~\ref{Cor:cyclicity} and Equation~(\ref{E:SumToIdentity}).

Notice that for every $y \in I(x)$ we have
\begin{equation}
\sum_i \lambda_i (A_y\otimes I_\cB)\ket{\psi_i} \bra{\psi_i}(A_y^*\otimes I_\cB) = (A_y\otimes I_\cB)\dfrac{\rho_x}{\tTr\rho_x}(A_y^*\otimes I_\cB) =\rho_y.
\end{equation} 
Next define $p_{i,y} =\bra{\psi_i}A_y^* A_y\otimes I_\cB\ket{\psi_i}$ and $\psi_{iy}$ by 
$\ket{\psi_{iy}} =\dfrac{ A_y \otimes I_\cB \ket{\psi_i}}{\sqrt{p_{iy} }}$
for each $i$ and each $y\in I(x)$ so that 
\begin{equation}
\sum_i \lambda_i \dfrac{p_{iy}}{\tTr\rho_y}\ket{\psi_{iy}}\bra{\psi_{iy}} = \frac{\rho_y}{\tTr\rho_y}.
\end{equation} 
 Now taking the trace of both sides of the above equation we get $\sum_i \dfrac{p_{iy}\lambda_i}{\tTr\rho_y} =1$. This implies that $\left\{ \dfrac{\lambda_i p_{iy}}{\tTr\rho_y}, \ket{\psi_{iy}} \right\}_{i\in I}$ is a pure state ensemble 
(\emph{not necessarily optimal}) of $\dfrac{\rho_y}{\tTr\rho_y}$ for each $y\in I(x)$. Thus
\begin{align}
\sum_i \lambda_i f\left( \tTr_\cA\sum_{y\in I(x)} (A_y\otimes I_\cB)\ket{\psi_i} \bra{\psi_i}(A_y^*\otimes I_\cB) \right) 
& = \sum_i \lambda_i f\left( \sum_y p_{iy} \tTr_\cA \ket{\psi_{iy}}\bra{\psi_{iy}} \right) \nonumber\\
& \geq \sum_{iy}\lambda_i p_{iy} f\left(\tTr_\cA\ket{\psi_{iy}}\bra{\psi_{iy}} \right) 
= \sum_{iy} \lambda_i p_{iy} \mu\left( \psi_{iy} \right)
\end{align}
Now since $\left\{ \dfrac{\lambda_i p_{iy}}{\tTr\rho_y}, \ket{\psi_{iy}} \right\}_{i\in I}$ is not necessarily an optimal decomposition of $\dfrac{\rho_y}{\tTr\rho_y}$, we arrive at the inequality 
\begin{equation}
\sum_i \dfrac{\lambda_i p_{iy}}{\tTr\rho_y} \mu\left( \psi_{iy} \right) \geq \mu\left(\dfrac{\rho_y}{\tTr\rho_y} \right)
\end{equation}
for each $y\in I(x)$.
It then follows that 
\begin{equation}
\sum_{iy} \lambda_i p_{iy} \mu\left( \psi_{iy} \right) = \sum_{iy} \lambda_i p_{iy}\dfrac{\tTr\rho_y}{\tTr\rho_y} \mu\left(\psi_{iy} \right)
= \sum_{y}\tTr\rho_y \sum_i \lambda_i \dfrac{p_{iy}}{\tTr\rho_y}\mu\left(\psi_{iy}\right)
\geq \sum_{y\in I(x)}\tTr\rho_y \mu\left( \dfrac{\rho_y}{\tTr\rho_y}\right)
\end{equation}
And therefore
\begin{equation}
\mu \left(\frac{\rho_x}{\tTr\rho_x} \right) \geq \sum_{y\in I(x)} \tTr\rho_y \mu \left(\dfrac{\rho_y}{\tTr\rho_y}\right)
\end{equation}
Thus $\mu$ monotonically decreases at each node in $\cT$. Next to show that $\mu(\rho)\geq \mu(\Lambda(\rho))$, we first iterate the argument over all nodes in the tree to obtain
\begin{equation}
\mu \left(\rho_{\emptyset}\right) 
\geq \sum_{y_1\in I(\emptyset)} \tTr\rho_{y_1} \mu\left(\frac{\rho_{y_1}}{\tTr\rho_{y_1}} \right) 
\geq \sum_{y_1\in I(\emptyset)} \tTr\rho_{y_1}\sum_{y_2 \in I(y_1)}\tTr\rho_{y_2}\mu \left( \dfrac{\rho_{y_2}}{\tTr\rho_{y_2}}\right) 
\geq \cdots 
\geq \sum_{y_1 \in I(\emptyset)}\tTr\rho_{y_1} \cdots \sum_{y_n \in I(y_{n-1})} \tTr \rho_{y_{n}} \mu\left( \dfrac{\rho_{y_n}}{\tTr\rho_{y_n}}\right)
\end{equation}
where each $\rho_{y_j}$ is the unnormalized state at node $y_j\in\cT$ and $y_n$ are the final nodes of $\cT$. Many 
authors \cite{HHHH,LOCC,MMO, VPRK,V98} consider the above inequality the defining quality of an LOCC monotone and 
say that $\mu$ decreases \emph{on average} under local operations and classical communication. But for a more 
functional result, we will go a step further and show that $\mu(\rho_\emptyset)\geq \mu(\Lambda(\rho_\emptyset))$. 
In this last step we will repeatedly use the fact that $\mu$ is convex which is an immediate consequence of the
definition of convex roof extension. This repeated use of the convexity of $\mu$ corresponds to loss of information
to the communicating parties, (Alice and Bob), and can be thought as black boxes in Alice's and Bob's labs
dismissing states produced by the LOCC channel without informing Alice or Bob, (see \cite[pp.~359-360]{V98}).
Continuing the previous argument, we use convexity in the nodes $y_n$ to obtain
\begin{equation}
 \sum_{y_1 \in I(\emptyset)}\tTr\rho_{y_1} \cdot \cdot \cdot \sum_{y_n \in I(y_{n-1})} \tTr \rho_{y_{n}} \mu\left( \dfrac{\rho_{y_n}}{\tTr\rho_{y_n}}\right) \geq \sum_{y_1 \in I(\emptyset)}\tTr\rho_{y_1} \cdot \cdot \cdot \sum_{y_{n-1} \in I(y_{n-2})} \tTr \rho_{y_{n-1}} \mu\left( \sum_{y_n \in I(y_{n-1})} \rho_{y_n}\right) 
\end{equation}

\noindent Now can assume without loss of generality that $\rho_{y_n}= \cE_{y_n}\left(\dfrac{\rho_{y_{n-1}}}{\tTr\rho_{y_{n-1}}}\right) $ for each $y_n\in I(y_{n-1} )$. Using this expression for $\rho_{y_n}$, it follows that 
\begin{equation}
\sum_{y_n \in I(y_{n-1})} \rho_{y_n} = \sum_{y_n\in I(y_{n-1})} \cE_{y_n}\left(\dfrac{\rho_{y_{n-1}}}{\tTr\rho_{y_{n-1}}}\right)
\end{equation}
is a \emph{normalized} state, which allows us to use convexity in the index $y_{n-1}$ so that
\begin{equation}
 \sum_{y_{n-1} \in I(y_{n-2})} \tTr \rho_{y_{n-1}} \mu\left( \sum_{y_n \in I(y_{n-1})} \rho_{y_n}\right) \geq \mu \left( \sum_{y_{n-1} \in I(y_{n-2})}\sum_{y_n \in I(y_{n-1})} \cE_{y_n}(\rho_{y_{n-1}}) \right)
\end{equation}
for each $y_{n-1}\in I(y_{n-2}) $. We then iterate the above argument to arrive at the inequality 
\begin{equation}
\mu(\rho_\emptyset)\geq \mu\left( \sum_{y_1 \in I(\emptyset)}\cdots \sum_{y_n\in I(y_{n-1})} \cE_{y_n}\circ\cdot\cdot\cdot\circ\cE_{y_1}(\rho_\emptyset)\right) = \mu \left(\Lambda (\rho_\emptyset) \right)
\end{equation}
\end{proof}
\end{widetext}

\begin{Cor} \label{Cor:e_entanglement_measure}
The entanglement number is an entanglement measure.
\end{Cor}

\begin{proof}
Instead of using the definitions of the entanglement number given in
Equations~(\ref{E:entanglement_number_for_pure_states_I}) and (\ref{E:entanglement_number_for_pure_states_II}),
we will use an alternative definition (\cite{G19, G20}) so that we can
invoke our proof for establishing Theorem~\ref{Thm:p_entanglement_measure} in the next section.

First define 
a function $f$ on density operators of finite dimensional Hilbert spaces by 
\begin{equation} \label{E:f}
 f(\rho)= \sqrt{1-\| \rho \|_2^2},
\end{equation}
where $\| \cdot \|_2$ denotes the Hilbert--Schmidt norm.
Then if $\cA \otimes \cB$ is a bipartite Hilbert space, and $\ket{\psi}$ is a pure state of $\cA \otimes \cB$, we can define the entanglement number by 
\begin{equation} \label{E:alternative_def_of_entangl_number}
 e(\psi)= f(\tTr_\cA \ket\psi\bra\psi)= f(\tTr_\cB \ket\psi\bra\psi) .
\end{equation}
Finally extend the definition of the entanglement number $e$ to all mixed states of $\cA\otimes \cB$ using the 
convex roof construction.

As Gudder proved, the entanglement number is faithful when restricted to pure states \cite{G19}. 
Moreover since the function $f$ is norm-continuous, the entanglement number is norm continuous on pure states by 
Equation~(\ref{E:alternative_def_of_entangl_number}).
Therefore by Corollary~\ref{Cor:faithful} the entanglement number is faithful, i.e. it vanishes only on 
separable states. 

In order to apply Theorem~\ref{Thm:LOCC_monotone} to the entanglement number, (and deduce that the entanglement
number is an LOCC monotone), notice that the function $f$ defined in Equation~(\ref{E:f}) is concave 
since the Hilbert--Schmidt norm is convex, 
the square function is increasing and
convex on the positive numbers, and the square root function is increasing and concave on the positive numbers.
Moreover $f(\rho)$ depends only on the singular numbers (and hence on the non-zero eigenvalues) of the density matrix $\rho$ for every density matrix $\rho$.
Furthermore, the extension of the entanglement number to the mixed states of $\cA \otimes \cB$ via the convex roof
function guarantees the existence of OPSE for all mixed states by Corollary~\ref{Cor:e_OPSE}. Thus by 
Theorem~\ref{Thm:LOCC_monotone}, the entanglement number is an LOCC monotone.

Finally notice that the entanglement number is invariant under local unitary transformations. Indeed for any
pure state $\ket\psi\bra\psi$ of a bipartite Hilbert space $\cA \otimes \cB$, local unitary transformations
will not effect the non-zero eigenvalues of $\tTr_\cA \ket\psi\bra\psi$, or of $\tTr_\cB \ket\psi\bra\psi$,
and therefore will not effect $e(\psi)$ according to 
Equation~(\ref{E:alternative_def_of_entangl_number}). This will remain valid under the convex roof extension of 
the entanglement number to mixed states.
\end{proof}

\section{p-Number of a State and Its Properties} \label{Sec:mu_p} 

Motivated by Equations~(\ref{E:f}) and (\ref{E:alternative_def_of_entangl_number}) we now define a family 
of entanglement measures. \emph{We assume that all Hilbert spaces 
mentioned in this section are finite dimensional}.
Let $\cZ$ be a (finite dimensional) Hilbert space and $1< p < \infty$ then
define $f_p:\tD(\cZ)\rightarrow \R_{\geq 0} $ by
\begin{equation}
f_p(\rho) = \left( 1- ||\rho||_p^p \right)^\frac{1}{p} 
\end{equation}
where $||\cdot||_p$ is the Schatten $p$-norm on $\tL(\cZ)$. 

\begin{Rmk}
For $1<p<\infty$ the function $f_p$ has the following properties:
\begin{enumerate}
\item $f_p$ depends only on the non-zero eigenvalues of its argument;
\item $f_p$ is concave;
\item $f_p$ is norm-continuous.
\end{enumerate}
\end{Rmk}

The proof of the Remark is similar to the corresponding properties of the function $f$ defined in Equation~(\ref{E:f}).

Then we define a function $\mu_p$ on pure states of a bipartite Hilbert space $\cA \otimes \cB$ by
\begin{equation}
 \mu_p(\psi) = f_p(\tTr_\cA \ket\psi \bra\psi) = f_p(\tTr_\cB \ket\psi \bra\psi)
\end{equation}
for all pure states $\ket\psi\in\cA\otimes\cB$, and extending $\mu_p$ by the convex roof construction 
to all mixed states as in Equation~(\ref{E:convexroof}). We call $\mu_p$ \emph{the $p$-number of a state}.
Now notice that for $p=2$, the $p$-number of a state coincides with the entanglement number of the state. 

\begin{Rmk}
For all $p\in(1,\infty)$, every mixed state on a bipartite Hilbert space $\cA \otimes \cB$ admits an OPSE for the convex roof construction
defining $\mu_p$.
\end{Rmk}

Indeed since $f_p$ is norm-continuous we obtain that $\mu_p$~is norm-continuous and the Remark follows from Theorem~\ref{Thm:Convex_Roof}.

\begin{Thm} \label{Thm:p_entanglement_measure}
For all $p\in(1,\infty)$, $\mu_p$ is an entanglement measure.
\end{Thm}

\begin{proof}
The proof of Corollary~\ref{Cor:e_entanglement_measure} repeats verbatim here, with the only addition that 
needs to be made is to verify that $\mu_p$ is faithful when restricted to pure states, (a statement that for 
the entanglement number $e$ was proved by Gudder~\cite{G19}).

For pure states, notice that we can compute $\mu_p$ using only the Schmidt decomposition of a state as follows. 
Let $\ket\psi\in\cA\otimes\cB$ be a pure state with Schmidt decomposition 
$\sum_k \sqrt{\lambda_k}\ket{\alpha_k}_\cA\otimes\ket{\beta_k}_\cB$. Then
\begin{align}
\tTr_\cA \ket\psi \bra\psi 
& = \tTr_\cA\left( \sum_{k,l} \sqrt{\lambda_k \lambda_l}\ket{\alpha_k}\bra{\alpha_l}\otimes \ket{\beta_k}\bra{\beta_l}\right) \nonumber\\
& =\sum_k \lambda_k \ket{\beta_k}\bra{\beta_k}
\end{align}
Notice that a similar result will follow if $\cB$ is traced out instead of $\cA$ because of the orthonormality of the $(\ket{\alpha_k})_k$ as well as the $(\ket{\beta_k})_k$.
Thus
\begin{equation} \label{E:mu_p_via_Schmidt}
\mu_p (\psi) = \left(1 - \sum_k \lambda_k^p \right)^\frac{1}{p}.
\end{equation}

Finally we verify that for a pure state $\psi$ in a bipartite Hilbert space $\cA \otimes \cB$, $\mu_p(\psi)=0$
if and only if $\psi$ is factorable.

$\implies)$ Suppose that $\mu_p(\psi)=0$ and let $\sum_k \sqrt{\lambda_k}\ket{\alpha}_\cA\otimes\ket{\beta_k}_\cB$ be the Schmidt decomposition of $\ket\psi$. Then the Schmidt coefficients satisfy the following system of equations.
\begin{equation}
\sum_k \lambda_k^p =1 \text{ and } \sum_k \lambda_k =1. 
\end{equation}
Since each $\lambda_k\in[0,1]$, it must follow that $\lambda_l=1$ for some $l$ and that $\lambda_j=0$ for all 
$j \neq l$. Thus $\ket\psi=\ket{\alpha_l}\otimes\ket{\beta_l}$ as desired. 

$\impliedby )$ Conversely, suppose that $\ket\psi = \ket u \otimes \ket v $ for some states $u\in\cA$ and $v\in\cB$. Then the length of the Schmidt decomposition is $1$ with Schmidt coefficient $1$. Thus $\mu_p (\psi)=0$.
\end{proof}

We also leave as an exercise to the reader to show that for any pure state $\ket\psi$ in a bipartite Hilbert space,
\begin{equation}\lim_{p\rightarrow 1} \frac{d}{dp} \log\mu_p(\psi) = \ent(\psi) \end{equation}

The p-numbers of a state also obey the following monotonicity property:
\begin{Rmk}
Let $1<p<q<\infty$, then $0<\mu_p(\rho)<\mu_q (\rho)<1$ for all states $\rho\in\tD(\cA\otimes\cB)$.
\end{Rmk}

This can be easily seen first for pure states using Equation~(\ref{E:mu_p_via_Schmidt}) and the fact that 
the sum of the $p$-th powers of the eigenvevalues of any density operator is in the interval $(0,1]$ for any 
$p \in (1,\infty)$, and then observing
that the monotonicity passes to the mixed states via the convex roof extension.

\begin{Rmk}
A measure of entanglement similar to the p-number was proposed by Cirone~\cite{MC}.
Cirone defined his measure $\nu_p$ only for pure bipartite states by 
\mbox{$\nu_p(\psi) =1- \sum_{i=1}^n \lambda_i^p$} where the $\sqrt{\lambda_i}$ are the Schmidt coefficients of 
the pure state $\psi$. Majorization techniques were used to discuss conversion of pure states to pure states
via LOCC operations and LOCC monotonicity was restricted to pure states only. 
It was left as an open problem to extend $\nu_p$ to general density matrices. The solution for states
on bipartite Hilbert spaces is implied by Gudder's \cite{G19} result on the entanglement number for 
the case of $p=2$, and our results in Section~\ref{Sec:mu_p} for all other $p$'s. Namely $\mu_p^p$ is 
the extension of $\nu_p$ to all bipartite states. 
\end{Rmk} 

\section{Concluding Remarks}

We have provided conditions under which the extension of any function from pure states to mixed states
via the convex roof construction will exhibit optimal pure state ensembles (OPSE). 
We applied this result to answer a question of Gudder about the existence of OPSE for the entanglement
number. Moreover we proved that the entanglement number is an entanglement measure. In order to prove that 
the entanglement number is LOCC monotone we used a criterion of Vidal. Furthermore we gave a simpler 
proof of Vidal's criterion and we represented LOCC channels using trees. This representation 
gives an interesting point of view on the LOCC operations. 
Finally we introduced a family of entanglement measures, the \mbox{$p$-numbers} of a bipartite state, 
parametrized by $p \in (1,\infty)$, such that the $2$-number of a state is equal to the 
entanglement number.

We do not know whether the entanglement number and the entanglement measures 
presented in Section~\ref{Sec:mu_p} can be extended to multipartite scenarios (i.e. tensor product of more 
that two Hilbert spaces).

\section*{Acknowledgement}

The work presented here is part of the PhD thesis of Ryan~McGaha which is conducted under 
the supervision of George Androulakis for the applied mathematics doctoral program at the University of 
South Carolina.


\begin{thebibliography}{10}
\expandafter\ifx\csname url\endcsname\relax
  \def\url#1{\texttt{#1}}\fi
\expandafter\ifx\csname urlprefix\endcsname\relax\def\urlprefix{URL }\fi
\expandafter\ifx\csname href\endcsname\relax
  \def\href#1#2{#2} \def\path#1{#1}\fi

\bibitem{HHHH}
R.~Horodecki, P.~Horodecki, M.~Horodecki, K.~Horodecki. Quantum entanglement.
  \emph{Reviews of Modern Physics} 2009; \textbf{81}(2):865--942.
\newblock \href {http://arxiv.org/abs/quant-ph/0702225}
  {\path{arXiv:quant-ph/0702225}}. \href
  {http://doi.org/10.1103/RevModPhys.81.865}
  {\path{doi:10.1103/RevModPhys.81.865}}.

\bibitem{G19}
S.~P. Gudder. A theory of entanglement. \emph{Quanta} 2020; \textbf{9}:7--15.
\newblock \href {http://doi.org/10.12743/quanta.v9i1.115}
  {\path{doi:10.12743/quanta.v9i1.115}}.

\bibitem{G20}
S.~P. Gudder. Quantum entanglement: spooky action at a distance. \emph{Quanta}
  2020; \textbf{9}:1--6.
\newblock \href {http://doi.org/10.12743/quanta.v9i1.113}
  {\path{doi:10.12743/quanta.v9i1.113}}.

\bibitem{PV05}
M.~B. Plenio, S.~Virmani. An introduction to entanglement measures.
  \emph{Quantum Information and Computation} 2007; \textbf{7}(1):1--51.
\newblock \href {http://arxiv.org/abs/quant-ph/0504163}
  {\path{arXiv:quant-ph/0504163}}.

\bibitem{VPRK}
V.~Vedral, M.~B. Plenio, M.~A. Rippin, P.~L. Knight. Quantifying entanglement.
  \emph{Physical Review Letters} 1997; \textbf{78}(12):2275--2279.
\newblock \href {http://arxiv.org/abs/quant-ph/9702027}
  {\path{arXiv:quant-ph/9702027}}. \href
  {http://doi.org/10.1103/PhysRevLett.78.2275}
  {\path{doi:10.1103/PhysRevLett.78.2275}}.

\bibitem{BBPSSW96}
C.~H. Bennett, G.~Brassard, S.~Popescu, B.~Schumacher, J.~A. Smolin, W.~K.
  Wootters. Purification of noisy entanglement and faithful teleportation via
  noisy channels. \emph{Physical Review Letters} 1996; \textbf{76}(5):722--725.
\newblock \href {http://arxiv.org/abs/quant-ph/9511027}
  {\path{arXiv:quant-ph/9511027}}. \href
  {http://doi.org/10.1103/PhysRevLett.76.722}
  {\path{doi:10.1103/PhysRevLett.76.722}}.

\bibitem{BDSW96}
C.~H. Bennett, D.~P. DiVincenzo, J.~A. Smolin, W.~K. Wootters. Mixed-state
  entanglement and quantum error correction. \emph{Physical Review A} 1996;
  \textbf{54}(5):3824--3851.
\newblock \href {http://arxiv.org/abs/quant-ph/9604024}
  {\path{arXiv:quant-ph/9604024}}. \href
  {http://doi.org/10.1103/PhysRevA.54.3824}
  {\path{doi:10.1103/PhysRevA.54.3824}}.

\bibitem{CNT87}
A.~Connes, H.~Narnhofer, W.~Thirring. Dynamical entropy of {$C^\star$} algebras
  and von {N}eumann algebras. \emph{Communications in Mathematical Physics}
  1987; \textbf{112}(4):691--719.
\newblock \href {http://doi.org/10.1007/bf01225381}
  {\path{doi:10.1007/bf01225381}}.

\bibitem{Uhl}
A.~Uhlmann. Optimizing entropy relative to a channel or a subalgebra.
  \emph{Open Systems and Information Dynamics} 1998; \textbf{5}:209--227.
\newblock \href {http://arxiv.org/abs/quant-ph/9701014}
  {\path{arXiv:quant-ph/9701014}}.

\bibitem{LCDJ}
S.~Lee, D.~P. Chi, S.~D. Oh, J.~Kim. Convex-roof extended negativity as an
  entanglement measure for bipartite quantum systems. \emph{Physical Review A}
  2003; \textbf{68}(6):062304.
\newblock \href {http://arxiv.org/abs/quant-ph/0310027}
  {\path{arXiv:quant-ph/0310027}}. \href
  {http://doi.org/10.1103/PhysRevA.68.062304}
  {\path{doi:10.1103/PhysRevA.68.062304}}.

\bibitem{GG}
G.~Gour. Family of concurrence monotones and its applications. \emph{Physical
  Review A} 2005; \textbf{71}(1):012318.
\newblock \href {http://arxiv.org/abs/quant-ph/0410148}
  {\path{arXiv:quant-ph/0410148}}. \href
  {http://doi.org/10.1103/PhysRevA.71.012318}
  {\path{doi:10.1103/PhysRevA.71.012318}}.

\bibitem{BL}
H.~Barnum, N.~Linden. Monotones and invariants for multi-particle quantum
  states. \emph{Journal of Physics A: Mathematical and General} 2001;
  \textbf{34}(35):6787--6805.
\newblock \href {http://arxiv.org/abs/quant-ph/0103155}
  {\path{arXiv:quant-ph/0103155}}. \href
  {http://doi.org/10.1088/0305-4470/34/35/305}
  {\path{doi:10.1088/0305-4470/34/35/305}}.

\bibitem{V98}
G.~Vidal. Entanglement monotones. \emph{Journal of Modern Optics} 2000;
  \textbf{47}(2-3):355--376.
\newblock \href {http://arxiv.org/abs/quant-ph/9807077}
  {\path{arXiv:quant-ph/9807077}}. \href
  {http://doi.org/10.1080/09500340008244048}
  {\path{doi:10.1080/09500340008244048}}.

\bibitem{LOCC}
E.~Chitambar, D.~Leung, L.~Man\v{c}inska, M.~Ozols, A.~Winter. Everything you
  always wanted to know about {LOCC} (but were afraid to ask).
  \emph{Communications in Mathematical Physics} 2014; \textbf{328}(1):303--326.
\newblock \href {http://arxiv.org/abs/1210.4583} {\path{arXiv:1210.4583}}.
  \href {http://doi.org/10.1007/s00220-014-1953-9}
  {\path{doi:10.1007/s00220-014-1953-9}}.

\bibitem{G20b}
S.~Gudder. Two entanglement measures. \emph{Journal of Physics: Conference
  Series} 2020; \textbf{1638}:012012.
\newblock \href {http://doi.org/10.1088/1742-6596/1638/1/012012}
  {\path{doi:10.1088/1742-6596/1638/1/012012}}.

\bibitem{Shirokov10}
M.~E. Shirokov. On properties of the space of quantum states and their
  application to the construction of entanglement monotones. \emph{Izvestiya:
  Mathematics} 2010; \textbf{74}(4):849--882.
\newblock \href {http://arxiv.org/abs/0804.1515} {\path{arXiv:0804.1515}}.
  \href {http://doi.org/10.1070/im2010v074n04abeh002510}
  {\path{doi:10.1070/im2010v074n04abeh002510}}.

\bibitem{G60}
J.~G. Glimm. On a certain class of operator algebras. \emph{Transactions of the
  American Mathematical Society} 1960; \textbf{95}(2):318--340.
\newblock \href {http://doi.org/10.1090/s0002-9947-1960-0112057-5}
  {\path{doi:10.1090/s0002-9947-1960-0112057-5}}.

\bibitem{FAN}
M.~Fannes. A continuity property of the entropy density for spin lattice
  systems. \emph{Communications in Mathematical Physics} 1973;
  \textbf{31}(4):291--294.
\newblock \href {http://doi.org/10.1007/bf01646490}
  {\path{doi:10.1007/bf01646490}}.

\bibitem{VW}
G.~Vidal, R.~F. Werner. Computable measure of entanglement. \emph{Physical
  Review A} 2002; \textbf{65}(3):032314.
\newblock \href {http://arxiv.org/abs/quant-ph/0102117}
  {\path{arXiv:quant-ph/0102117}}. \href
  {http://doi.org/10.1103/PhysRevA.65.032314}
  {\path{doi:10.1103/PhysRevA.65.032314}}.

\bibitem{DE}
D.~Eisenbud. Commutative Algebra with a View Toward Algebraic Geometry.
  Springer, New York, 1995.
\newblock \href {http://doi.org/10.1007/978-1-4612-5350-1}
  {\path{doi:10.1007/978-1-4612-5350-1}}.

\bibitem{OHYA}
M.~Ohya, D.~Petz. Quantum Entropy and Its Use. Texts and Monographs in Physics.
  Springer, Berlin, 1993.

\bibitem{A17}
G.~Aubrun, S.~Szarek. Alice and Bob Meet Banach: The Interface of Asymptotic
  Geometric Analysis and Quantum Information Theory. Mathematical Surveys and
  Monographs, 2017.
\newblock \href {http://doi.org/https://doi.org/10.1090/surv/223}
  {\path{doi:https://doi.org/10.1090/surv/223}}.

\bibitem{MMO}
M.~J. Donald, M.~Horodecki, O.~Rudolph. The uniqueness theorem for entanglement
  measures. \emph{Journal of Mathematical Physics} 2002;
  \textbf{43}(9):4252--4272.
\newblock \href {http://arxiv.org/abs/quant-ph/0105017}
  {\path{arXiv:quant-ph/0105017}}. \href {http://doi.org/10.1063/1.1495917}
  {\path{doi:10.1063/1.1495917}}.

\bibitem{MC}
M.~A. Cirone. Quantifying entanglement with probabilities 2001; \href
  {http://arxiv.org/abs/quant-ph/0110139} {\path{arXiv:quant-ph/0110139}}.

\end{thebibliography}
\end{document}